%% file: root.tex
\documentclass[11pt]{article}

\usepackage{mathrsfs}
\usepackage{amsmath}
\usepackage{amssymb}
\usepackage{amsthm}
\usepackage{latexsym}
\usepackage[dvips]{graphicx,color}
\usepackage{float}
\usepackage{array}
\usepackage{xcolor}
\usepackage{enumerate} 
\usepackage[shortlabels]{enumitem}
\usepackage{titlesec}
\usepackage{multirow}
\usepackage[hyphens]{url}
\usepackage{hyperref}
\usepackage{etoolbox}
\usepackage{xspace}
\usepackage[normalem]{ulem}

\titleformat{\section}[block]{\Large\sc\filcenter}{\thesection.}{5pt}{}
\titleformat{\subsection}[block]{\sc\filcenter}{\thesubsection.}{5pt}{}

\makeatletter
\providecommand{\institute}[1]{% add institute to \maketitle
  \apptocmd{\@author}{\end{tabular}
    \par\smallskip
    \begin{tabular}[t]{c}
    #1}{}{}
}
\providecommand{\email}[1]{% add email to \maketitle
  \apptocmd{\@author}{\end{tabular}
    \par\smallskip
    \begin{tabular}[t]{c}
    \texttt{#1}}{}{}
}
\makeatother

\theoremstyle{plain}
\newtheorem{theorem}{Theorem}[section]

\newtheorem{lemma}[theorem]{Lemma}

\theoremstyle{definition}

\newtheorem{definition}[theorem]{Definition}

\theoremstyle{remark}
      
\newtheorem{claim}{Claim}

\input{macros}

\begin{document}

\title{Extension Preservation in the Finite and Prefix Classes of First Order Logic}

\author{\large{Anuj Dawar and Abhisekh Sankaran}\\[3pt]}

\institute{\small{Department of Computer Science and Technology,}\\\small{University of Cambridge, U.K.}}
\email{\small{\{anuj.dawar, abhisekh.sankaran\}@cl.cam.ac.uk}}
\date{}
\maketitle

\begin{abstract}
\input{abstract}
\end{abstract}

\input{intro}

\input{preliminaries}

\input{property}
\input{proof}

\input{conclusion}

%%%%%%%%%%%%%%%%%%%%%%%% Bibliography %%%%%%%%%%%%%%%%%%%%%%%%%%%%%%

\bibliographystyle{plain}

\bibliography{refs}

%%%%%%%%%%%%%%%%%%%%% ---------------- %%%%%%%%%%%%%%%%%%%%%%%%%%%%%

\end{document}

%% file: macros.tex
\newcommand{\fo}{\ensuremath{\text{FO}}}

\newcommand{\efgame}{Ehrenfeucht-Fra\"iss\'e\xspace}

\newcommand{\str}[1]{\ensuremath{\mathfrak{#1}}}

\newcommand{\ra}{\rightarrow}

\newcommand{\predicate}[1]{\textrm{#1}}

\newcommand{\Succ}{\predicate{Succ}}

\newcommand{\Total}{\predicate{Total}}
\newcommand{\PartialSucc}{\predicate{PartialSucc}}
\newcommand{\NotPartialSucc}{\predicate{NotPartialSucc}}
\newcommand{\SomeTotalR}{\predicate{SomeTotalR}}
\newcommand{\RTotal}{\predicate{RTotal}}
\newcommand{\Rule}[2]{#1 & \longleftarrow & #2}
\newcommand{\Datalog}{\ensuremath{\texttt{Datalog}(\neg)}\xspace}
\newcommand{\nats}{\mathbb{N}}
\newcommand{\Tot}{\mathsf{Tot}}
\newcommand{\Gap}{\mathsf{Gap}}

\newcommand{\LT}{{\L}o\'s-Tarski\xspace}
\renewcommand{\phi}{\varphi}
\newcommand{\tup}[1]{\bar{#1}}

%% file: abstract.tex
It is well known that the classic {\L}o\'s-Tarski preservation theorem fails in the finite: there are first-order definable classes of finite structures closed under extensions which are not definable (in the finite) in the existential fragment of first-order logic.  We strengthen this by constructing for every $n$, first-order definable classes of finite structures closed under extensions which are not definable with $n$ quantifier alternations.  The classes we construct are definable in the extension of Datalog with negation and indeed in the existential fragment of transitive-closure logic.  This answers negatively an open question posed by Rosen and Weinstein.

%% file: intro.tex
\section{Introduction}\label{section:intro}
The failure of classical preservation theorems of model theory has
been a topic of persistent interest in finite model theory.  In the
classical setting, preservation theorems provide a tight link between
the syntax and semantics of first-order logic ($\fo$).  For instance, the \LT
preservation theorem (see~\cite{hodges}) implies that any sentence of first-order logic
whose models are closed under extensions is equivalent to an
existential sentence.  This, like many other classical preservation
theorems, is false when we retrict ourselves to finite structures.
Tait~\cite{tait} and Gurevich~\cite{gurevich} provide examples of
sentences whose finite models are closed under extensions, but which
are not equivalent, over finite structures, to any existential
sentence.  Many other classical preservation theorems have been
studied in the context of finite model theory
(e.g.~\cite{Rosen,Rossman}), but our focus in this paper is on
extension-closed properties.

The failure of the \LT theorem in the finite opens a number of
different avenues of research.  One line of work has sought to
investigate restricted classes of structures on which a version of the
preservation theorem holds (see~\cite{ADG,tame}).  Another direction
is prompted by the question of whether we can identify some proper
syntactic fragment of $\fo$, beyond the existential, which contains
definitions of all extension-closed $\fo$-definable properties.  For
instance, the examples from Tait and Gurevich are both $\Sigma_3$ sentences.
Could it be that every $\fo$ sentence whose finite models are
closed under extensions is equivalent to a $\Sigma_3$ sentence?  Or,
indeed, a $\Sigma_n$ sentence for some constant $n$?  We answer these
questions negatively in this paper.  That is, we show that we can
construct, for each $n$, a sentence $\phi$ whose finite models are
extension closed but which is not equivalent in the finite to a $\Sigma_n$
sentence.  

A related question is posed by Rosen and Weinstein~\cite{RosenW}.
They observe that the constructions due to Tait and Gurevich 
both yield classes of finite structures that are definable
in $\texttt{Datalog}(\neg)$, the existential fragment of
fixed-point logic in which only extension-closed properties can be expressed.
They ask if it might be the case that $\fo \cap
\texttt{Datalog}(\neg)$ is contained in some level of the
first-order quantifier alternation hierarchy, be it not the lowest level.  That is, could it be that every property that is first-order definable and also definable in $\Datalog$ is definable by a $\Sigma_n$ sentence for some constant $n$?
Our
construction answers this question negatively as we show that the
sentences  we construct are all equivalent to formulas of
$\texttt{Datalog}(\neg)$.

Our result also greatly strengthens a previous result by Sankaran~\cite{sankaran19} which showed that for each $k$ there is an extension-closed property of finite structures definable in $\fo$ but not in $\Pi_2$ with $k$ leading universal quantifiers.  Indeed, our result answers (negatively) Problem~2 in~\cite{sankaran19}.

In Section~\ref{section:preliminaries} we give the necessary
background definitions.  We construct the sentences in
Section~\ref{section:the-property-psi-n} and show that they can all be
expressed in $\Datalog$.  Section~\ref{section:proof} contains the
proof that the sequence of sentences contains, for each $n$, a sentence that is not equivalent to any $\Sigma_n$ sentence.  We conclude with some suggestions for further investigation.

%% file: preliminaries.tex
\section{Preliminaries}\label{section:preliminaries}
We work with logics: first-order logic ($\fo$) and extensions of
$\texttt{Datalog}$ over finite relational vocabularies.  We assume the
reader is familiar with the basic definitions of first-order logic (see, for instance~\cite{libkin}).  A
\emph{vocabulary} $\tau$ is a set of predicate and constant symbols.
In the vocabularies we use, all predicate symbols are either unary or
binary.  We denote by $\fo(\tau)$ the set of all $\fo$ formulas over
the vocabulary $\tau$.  A sequence $(x_1, \ldots, x_k)$ of variables
is denoted by $\bar{x}$.  We use $\psi(\bar{x})$ to denote a formula
$\psi$ whose free variables are among $\bar{x}$.  A formula without
free variables is called a \emph{sentence}.  A formula which begins
with a string of quantifiers that is followed by a quantifier-free
formula, is said to be in \emph{prenex normal form (PNF)}.  The string
of quantifiers in a PNF formula is called the \emph{quantifier prefix}
of the formula.  It is well known that every formula is equivalent to
a formula in PNF.  We denote by $\Sigma_n$, the collection of all
formulas in PNF whose quantifier prefix contains at most $n$ blocks of
quantifiers beginning with a block of existential quantifiers.
Equivalently, a PNF formula is in $\Sigma_n$ if it starts with a block of
existential quantifiers and contains at most $n-1$ alternations in its
quantifier prefix.  Similarly, a formula is $\Pi_n$ if it begins with
a block of universal quantifiers and contains at most $n-1$
alternations in its quantifier prefix.  We write $\Sigma_{n, k}$ for
the subclass of $\Sigma_n$ consisting of those formulas in which every
quantifier block has at most $k$ quantifiers.  Similarly, $\Pi_{n, k}$
is the subclass of $\Pi_n$ where each block has at most $k$
quantifiers.  Thus $\Sigma_n = \bigcup_{k \ge 1} \Sigma_{n, k}$ and
$\Pi_n = \bigcup_{k \ge 1} \Pi_{n, k}$.

We use standard notions concerning $\tau$-structures as defined in~\cite{chang-keisler-short}.  We denote $\tau$-structures as $\str{A},
\str{B}$ etc., and refer to them simply as structures when $\tau$ is
clear from the context.  We denote by $\str{A} \subseteq \str{B}$ that
$\str{A}$ is a substructure of $\str{B}$, and by $\str{A} \cong
\str{B}$ that $\str{A}$ is isomorphic to $\str{B}$. 

We now introduce some notation with respect to the classes of formulas
$\Sigma_{n,k}$ and $\Pi_{n,k}$.
\begin{definition}\label{def:sigmank-relation}
  We say $\str{A} \Rrightarrow_{n, k} \str{B}$ if  every $\Sigma_{n, k}$
sentence true in $\str{A}$ is also true in $\str{B}$.

We say $\str{A}$ and $\str{B}$ are \emph{$\equiv_{n, k}$-equivalent},
and write $\str{A} \equiv_{n, k} \str{B}$, if  $\str{A} \Rrightarrow_{n, k}
\str{B}$ and $\str{B} \Rrightarrow_{n, k} \str{A}$.
\end{definition}
By extension, for tuples $\bar{a}$ and $\bar{b}$ of elements of $\str{A}$ and $\str{B}$ respectively, we also write $(\str{A},\bar{a}) \Rrightarrow_{n,k} (\str{B},\bar{b})$ to indicate that every formula $\phi$ which is satisfied in $\str{A}$ when its free variables are instantiated with $\bar{a}$ is also satisfied in $\str{B}$ when they are instantiated by $\bar{b}$, and similarly for $\equiv_{n,k}$.
Note that $\str{A} \Rrightarrow_{n, k} \str{B}$ holds if every $\Pi_{n, k}$
sentence true in $\str{B}$ is also true in $\str{A}$.  The following useful fact is now immediate from the definition.
\begin{lemma}\label{lem:sigmank-ind}
 $\str{A} \Rrightarrow_{n+1, k} \str{B}$ if, and only if, for every $k$-tuple $\bar{a}$ of elements of $\str{A}$, there is a $k$-tuple $\bar{b}$ of elements of $\str{B}$ such that $(\str{B},\bar{b}) \Rrightarrow_{n,k} (\str{A},\bar{a})$.
\end{lemma}

We
assume the reader is familiar with the standard {\efgame} game
characterizing the equivalence of two structures with respect to
sentences of a given quantifier nesting depth (see for
example~\cite[Chapter~3]{libkin}).  In this
paper, we use a ``prefix'' variant of this \efgame
game.  For $n, k \ge 1$, the \emph{$(n, k)$-prefix \efgame game} on a pair
$(\str{A}, \str{B})$ of structures, is the usual \efgame game on $\str{A}$
and $\str{B}$ but with two restrictions: (i) In every odd
round, the Spoiler plays on $\str{A}$ and in every even round, 
on $\str{B}$, and (ii) in
each round the Spoiler chooses a $k$-tuple of elements from the
relevant structure (as opposed to a single element in the usual
\efgame game).  The winning condition for the Duplicator 
is the same as that in the usual \efgame game: Duplicator wins at the end of $n$
rounds, if when $\bar{a}_1, \ldots, \bar{a}_n$ are the $k$-tuples picked in
$\str{A}$ and $\bar{b}_1, \ldots, \bar{b}_n$ are the $k$-tuples picked
in $\str{B}$, then the map taking the $nk$-tuple $(\bar{a}_1\cdots\bar{a}_n)$ to $(\bar{b}_1\cdots\bar{b}_n)$ pointwise is a partial isomorphism from $\str{A}$ to $\str{B}$.  Entirely analogously to the
usual \efgame game theorem (see~\cite[Theorem~3.18]{libkin}), we have the
following. 
\begin{theorem}\label{thm:EF}
  Duplicator has a winning strategy in the \emph{$(n, k)$-prefix \efgame game} on a pair
$(\str{A}, \str{B})$ of structures if, and only if, $\str{A} \Rrightarrow_{n, k}
\str{B}$
\end{theorem}
Note in particular that, given the fact that any
two linear orders of length $\ge 2^n$ are equivalent with respect to
all sentences of quantifier nesting depth $n$~\cite[Theorem~3.6]{libkin}, it
follows that there is a winning strategy for the Duplicator in the
$(n, k)$-prefix \efgame game on any pair of linear orders, each of length
$\ge 2^{n \cdot k}$ and any two such linear orders are
$\equiv_{n, k}$-equivalent.

Where it causes no confusion, we still use $\equiv_m$ to denote the usual equivalence up to \emph{quantifier rank} $m$.  Note that $\str{A} \equiv_m \str{B}$ implies $\str{A} \equiv_{n,k} \str{B}$ whenever $m \geq nk$.

Formulas in $\Sigma_1$ are said to be \emph{existential}.  A $\Sigma_1$ formula that also contains no occurrences of the negation symbol is said to be \emph{existential positive}.  It is easy to see that the class of models of any $\Sigma_1$ sentence $\phi$ is closed under extensions: if $\str{A} \models \phi$ and $\str{A} \subseteq \str{B}$, then $\str{B} \models \phi$.  Dually, the class of models of any $\Pi_1$ sentence is closed under taking substructures.  Similarly, the class of models of any existential positive sentence is closed under homomorphisms.

$\texttt{Datalog}$ is a database query language which can be seen as
an extension of existential positive first-order logic with a
recursion mechanism.  Equivalently, it can be seen as the
existential positive fragment of the logic of least fixed points
$\texttt{LFP}$ (see~\cite[Chapter~10]{libkin}).   We briefly review the definitions of this language, along with its extension $\texttt{Datalog}(\neg)$.

A \emph{Datalog program} is a finite set of rules of the form $T_0 \leftarrow T_1,\ldots,T_m$, where each $T_i$ is an atomic formula.  $T_0$ is called the \emph{head} of the rule, while the right-hand side is called the \emph{body}.  These atomic formulas use relational symbols from a vocabulary $\sigma \cup \tau$, where the symbols in $\sigma$ are called \emph{extensional} predicates and those in $\tau$ are \emph{intensional} predicates.  Every symbol that occurs in the head of a rule is an intensional predicate, while both intensional and extensional predicates can occur in the body of a rule.  The semantics of such a program is defined with respect to a $\sigma$-structure $\str{A}$.  Say that a rule $T_0 \leftarrow T_1,\ldots,T_m$ is satisfied in a $\sigma \cup \tau$ expansion $\str{A}'$ of $\str{A}$ if $\str{A}' \models \forall \tup{x} \big((\bigwedge_{1\leq i \leq m} T_i) \rightarrow T_0\big)$, where $\tup{x}$ enumerates all the variables occurring in the rule.  The interpretation of a $\texttt{Datalog}$ program in $\str{A}$ is the smallest expansion of $\str{A}$ (when ordered by pointwise inclusion of the relations interpreting $\tau$) satisfying all the rules in the program.  This is uniquely defined as it is obtained as the simultaneous least fixed-point of the existential closure of the right-hand side of the rules.  We distinguish one intensional predicate $G$ and call it the \emph{goal predicate}.  Then, the query computed by a program $\pi$ is the interpretation of $G$ in the interpretation of $\pi$ in $\str{A}$.  In particular, if $G$ is a $0$-ary predicate symbol (i.e.\ a Boolean variable), $\pi$ defines a Boolean query,  i.e.\ a class of structures.

Since the interpretation of  $\pi$ is obtained as the least fixed-point of an existential positive formula, it is easily seen that the query defined is closed under homomorphisms and hence also under extensions.  We can understand $\texttt{Datalog}$ as the existential positive fragment of the least-fixed point logic LFP, though it is known that there are homomorphism-closed properties definable in LFP that are not expressible in $\texttt{Datalog}$ (see~\cite{DK08}).

We get more general queries by allowing limited forms of negation.  Specifically, in $\texttt{Datalog}(\neg)$, in a rule $T_0 \leftarrow T_1,\ldots,T_m$, each $T_i$ on the right-hand side is either an atom \emph{or} a negated atom involving an \emph{extensional} predicate symbol or equality.  In short, we allow negation on the predicate symbols in $\sigma$ and on equalities but the fixed-point variables (i.e.\ the predicate symbols in $\tau$) still only appear positively, so the least fixed-point is still well defined.  As it is still the least-fixed point of existential formulas, the formula still defines a property closed under extensions.   For more on the extensions of $\texttt{Datalog}$ with negation, see~\cite{AHV-book,KV-datalog}.

%% file: property.tex
\section{The Extension-Closed Properties}\label{section:the-property-psi-n}

We now construct a family of properties, each of which is definable
in first-order logic and closed under extensions.  Indeed, we show
that each of the properties is definable in $\Datalog$. 

\subsection{First-Order Definitions}
To begin, we define, for each $n \in \nats$, a vocabulary $\sigma_n$.
These are defined by induction on $n$.  The vocabulary $\sigma_1$
consists of three binary relation symbols $\leq,S,R$.  For all $n >
1$, $\sigma_{n} = \sigma_{n-1} \cup \{S_n,R_n,P_n\}$ where $S_n$ and
$R_n$ are binary relation symbols and $P_n$ is a unary relation
symbol.

Consider first the sentence $\predicate{NLO}$ of $\fo$ which asserts
that $\leq$ is \emph{not} a linear order.  This is easily seen to be an existential
sentence and so also definable in $\Datalog$.  Suppose now that $\phi$
is any sentence whose models, restricted to ordered structures (i.e.\ those structures which
interpret $\leq$ as a linear order), are extension closed.  Then, it
follows that $\predicate{NLO} \lor \phi$ defines an extension-closed
class of structures.  Moreover, this class is $\fo$ or $\Datalog$
definable if $\phi$ is in the respective logic.  Also, if $\phi$ is a $\Sigma_n$ sentence, then so is $\predicate{NLO} \lor \phi$, and if $n> 1$ and $\phi$ is a $\Pi_n$ sentence then so is $\predicate{NLO} \lor \phi$.  Thus, in what
follows, we restrict our attention to the class of ordered structures.
We construct our sentences on the assumption that structures are
ordered, and show that they define extension-closed classes on ordered
structures. 

With this in
mind, we use some convenient notational abbreviations.  We write $x <
y$ as short-hand for $x \leq y \land x \neq y$.  We also write ``$y$ is
the successor of $x$'', ``$x$ is the minimum element'', etc.\ with
their obvious meanings.  Also, let $\phi$ be any formula of $\fo$, and
$x$ and $y$ be variables not occuring in $\phi$.  We write
$\phi^{[x,y]}$ for the formula $x \leq y \land \phi^{\star}$ where
$\phi^{\star}$ is the formula obtained by relativizing every quantifier
in $\phi$ to the interval $[x,y]$.  That is to say, inductively, every
subformula $\exists z \theta$ is replaced by $\exists z( x \leq z
\land z \leq y) \land \theta^{\star}$ and every subformula  $\forall z
\theta$ by  $\forall z( x \leq z \land z \leq y) \ra \theta^{\star}$.
Where the variables $x$ and $y$ do appear in $\phi$, the formula $\phi^{[x,y]}$ is
defined by first renaming variables in $\phi$ to avoid clashes and then
applying the relativization.

Next, consider the sentence $\PartialSucc$ defined as
follows. 
\[\PartialSucc \; := \;  \forall x \forall y ~ S(x,y) \rightarrow \text{``$y$ is the successor of $x$''}. \]
This is a $\Pi_1$ sentence, asserting that the relation $S$ is a
``partial successor'' relation.  Its negation is an existential
sentence and hence closed under extensions.  Thus, if $\phi$ defines
an extension-closed class of structures when restricted to ordered
structures in which $S$ is a partial successor relation, then
$\predicate{NLO} \lor \neg \PartialSucc \lor \phi$ defines
an extension-closed class. 

We write $\Total(x,y)$ for the formula that asserts, on structures in which $\PartialSucc$ is true, that in the
interval $[x,y]$, $S$ is, in fact, total.  That is:
\[\Total(x,y) \; := \; x < y \land \forall z \big(x \leq z \wedge z <
    y\big) \rightarrow \exists w \big(z < w \wedge w \leq y \wedge
    S(z, w) \big). \]

Now, we can define the sentence $\SomeTotalR_1$.
\[ \SomeTotalR_1 \; := \; \neg \PartialSucc \lor \exists x \exists y \, \big(R(x,y) \land
\Total(x,y)\big).\]

Note that $\Total(x,y)$ is a $\Pi_2$ formula, and $\SomeTotalR_1$
is a $\Sigma_3$ sentence.  The latter is the first in our family of
sentences.  To see why this sentence is closed under extensions on ordered structures, suppose $\str{A}$ is an ordered model of $\SomeTotalR_1$ on which $S$ is a partial successor.   Thus, there is an interval $[x,y]$ in $\str{A}$ on which $S$ is total.  Let $\str{B}$ be an extension of $\str{A}$.  If $\str{B}$ contains no additional elements in the interval $[x,y]$, then $\Total(x,y)$ still holds in $\str{B}$ and therefore $\str{B}$ is a model of $\SomeTotalR_1$.  On the other hand, suppose $\str{B}$ contains an additional element $w$ in the interval $[x,y]$.  Let $a$ and $b$ be the two successive elements of $\str{A}$ between which $w$ appears.  Since $S$ is total in the interval $[x,y]$ in $\str{A}$, we know that $S(a,b)$ holds in $\str{A}$ and, by extension, in $\str{B}$.  Since $b$ is not the successor of $a$ in $\str{B}$, we conclude that $\neg \PartialSucc$ is true in $\str{B}$ and therefore the structure is a model of $\SomeTotalR_1$.

The sentence $\SomeTotalR_1$ is essentially the example constructed by Tait that exhibits an existential-closed first-order property that is not expressible by an existential sentence.
We now define $\sigma_n$-sentences $\SomeTotalR_n$, for
$n> 0$ by induction.

First, we define a formula $\Succ_n(x,y)$ as follows.
\[ \Succ_n(x,y) \; := \; P_n(x) \land P_n(y) \land S_n(x, y) \land \SomeTotalR_{n-1}^{[x,y]}. \]
We further define the formula $\PartialSucc_n$ which asserts that
$\Succ_n$ is a partial successor relation when restricted to the
elements in the relation $P_n$.  That is, 
\[ \PartialSucc_n \; := \; \forall x \forall y ~ \Succ_n(x,y) \ra
\forall z ( P_n(z) \ra z \leq x \lor y \leq z). \]

We can now define the formula $\Total_n(x,y)$ which defines, in those
structures in which $\PartialSucc_n$ is true, those intervals $[x,y]$
where the successor defined by $\Succ_n$ is total.  That is, 
\[\Total_n(x,y) \; := \; x < y \land \forall z \big(P_n(z) \land x \leq z \wedge z <
    y\big) \rightarrow \exists w \big(z < w \wedge w \leq y \wedge
    \Succ_n(z, w) \big). \]

Finally, we define the sentence 
\[ \SomeTotalR_n \; := \; \neg \PartialSucc_n \lor \exists x \exists
y \, \big( R_n(x,y) \land \Total_n(x,y) \big).\]

Note that, $\SomeTotalR_n$ is a $\Sigma_{2n+1}$ sentence.  This can be
established by induction on $n$.  Indeed, as we noted, $\SomeTotalR_1$
is a $\Sigma_3$ sentence.  Assuming $\SomeTotalR_n$ is a $\Sigma_{2n+1}$
sentence for some $n$, we note that $\Succ_{n+1}$ is a
$\Sigma_{2n+1}$ formula, and so is $\neg \PartialSucc_{n+1}$.  Then
$\Total_{n+1}$ is a $\Pi_{2n+2}$ formula and $\SomeTotalR_{n+1}$ is
$\Sigma_{2n+3}$. 

\subsection{Datalog Definitions}
Next, we show that these formulas also admit a definition in
$\Datalog$, which establishes, in particular, that they define
extension-closed classes.  We use the same names for formulas in
$\Datalog$ as we used for $\fo$ formulas above, when they define the
same property.  As we noted, the sentences $\predicate{NLO}$ and
$\neg\PartialSucc$ are both $\Sigma_1$ sentences and we therefore
assume they are available as $\Datalog$ predicates.  We now define
$\Total$ by the following rules.
\[
\begin{array}{l@{\hspace{2mm}}c@{\hspace{2mm}}l}
 \Rule{\Total(x,y)}{S(x,y)} \\
 \Rule{\Total(x,y)}{S(x,z),\Total(z,y)}
\end{array}
\]
This just defines $\Total$ as the transitive closure of $S$.  It
is clear that, in ordered structures where $S$ is a partial successor
relation, the pair $(x,y)$ is in the transtive closure of $S$ precisely when $x
< y$ and $S$ is total in the interval $[x,y)$.  Thus, we can now
define:
\[
\begin{array}{l@{\hspace{2mm}}c@{\hspace{2mm}}l}
 \Rule{\RTotal_1(x,y)}{x\leq u, v \leq y, R(u,v),\Total(u,v)} 
\end{array}
\]
This defines those pairs $(x,y)$ such that for some $u,v$ in the
interval $[x,y]$, $R(u,v)$ holds and the successor relation is total.
In other words, it defines $\SomeTotalR_1^{[x,y]}$.  We can obtain
$\SomeTotalR_1$ as the existential closure of this.  For the inductive
definition, the predicate $\RTotal_n$ is useful.

Inductively, we define the relation, $\Succ_n$ as follows.
\[
\begin{array}{l@{\hspace{2mm}}c@{\hspace{2mm}}l}
 \Rule{\Succ_{n}(x,y)}{P_n(x),P_n(y),S_n(x, y),\RTotal_{n-1}(x,y)}
\end{array}
\]
The negation of $\PartialSucc_n$ is now defined by the following
\[
\begin{array}{l@{\hspace{2mm}}c@{\hspace{2mm}}l}
 \Rule{\NotPartialSucc_{n}}{\Succ_n(x,y),P_n(z), x\leq z, z\leq y, x\neq
  z, y\neq z}\\
\end{array}
\]

Now, entirely analogously to $\Total$ above, we can give a definition of
$\Total_n$ as the transitive closure of $\Succ_n$ and this is
equivalent to the $\fo$ definition given above on ordered structures
on which $\PartialSucc_n$ is true.
\[
\begin{array}{l@{\hspace{2mm}}c@{\hspace{2mm}}l}
 \Rule{\Total_{n}(x,y)}{\Succ_{n}(x,y)} \\
 \Rule{\Total_{n}(x,y)}{\Succ_{n}(x,z),\Total_{n}(z,y)}
\end{array}
\]
Inductively we define the relation $\RTotal_n$, and its existential
closure, giving the sentence $\SomeTotalR_n$.
\[
\begin{array}{l@{\hspace{2mm}}c@{\hspace{2mm}}l}
 \Rule{\RTotal_n(x,y)}{x\leq u, v \leq y, R_n(u,v),\Total_n(u,v)}
\end{array}
\]

It should be noted that the only use of the recursive features of $\Datalog$ that we made use in writing the formulas above was to define the transitive closure of the relations $\Total$  and $\Total_n$.  Thus, the definitions could equally well be formalized in the existential fragment of transitive closure logic. 

%% file: proof.tex
\section{Proof of the Main Result}\label{section:proof}

In this section, we establish our main result.  We establish that $\SomeTotalR_n$, which we noted is a $\Sigma_{2n+1}$ sentence, is not
equivalent to a $\Pi_{2n+1}$ sentence.   To do this, we construct ordered structures $\str{M}_{n,k}$ and $\str{N}_{n,k}$ for every $k$ such that $\str{M}_{n,k}$ is a model of  $\SomeTotalR_n$, $\str{N}_{n,k}$ is not a model of  $\SomeTotalR_n$ but $\str{N}_{n,k} \Rrightarrow_{2n+1,k} \str{M}_{n,k}$.  The main lemma establishing this is Lemma~\ref{lem:main} below.  Here we state the theorem that is a consequence.
\begin{theorem}\label{thm:main}
  For every $n$, there is a $\Sigma_{2n+1}$ sentence whose finite models are closed under extensions and which is equivalent to a $\Datalog$ program, but which is not equivalent over finite structures to any $\Pi_{2n+1}$ sentence.
\end{theorem}
\begin{proof}
  The sentence is  $\predicate{NLO} \vee \SomeTotalR_n$ which we have already noted is a $\Sigma_{2n+1}$ sentence, expressible as a $\Datalog$ program and its models are extension-closed.  Suppose it were expressible as a  $\Pi_{2n+1}$ sentence.  Then, since it is satisfied in $\str{M}_{n,k}$ as we show in Section~\ref{sec:construct} and since $\str{N}_{n,k} \Rrightarrow_{2n+1,k} \str{M}_{n,k}$ by Lemma~\ref{lem:main} we have that the sentence is true in $\str{N}_{n,k}$.   But, as we show in Section~\ref{sec:construct}, $\str{N}_{n,k}$ is not a model of $\SomeTotalR_n$, yielding a contradiction.
\end{proof}

\input{construct3}

\input{game}

%% file: construct3.tex
\subsection{Construction of the Structures}\label{sec:construct}

We describe the construction of structures $\str{M}_{n,k}$ and $\str{N}_{n,k}$ for each $n$ and $k$.  The construction is by induction on $n$, simultaneously for all $k$.   In the course of the construction we also define, for all $n$ and $k$ structures $\Tot_{n,k}$ and $\Gap_{n,k}$ which we use as auxilliary structures.  For all $n$ and $k$,  $\str{M}_{n,k},\str{N}_{n,k}, \Tot_{n,k}$ and $\Gap_{n,k}$ are structures over the vocabulary $\sigma_n$.

All structures we consider interpret the relation symbol $\leq$ as a linear order of the universe and $S$ as a partial successor relation.  It is useful to formally define the notion of an ordered sum of
structures.  For a pair $\str{A}$ and $\str{B}$ of ordered structures the \emph{ordered sum} $\str{A} \oplus \str{B}$ is a
structure whose universe is the disjoint union of the universes of
$\str{A}$ and $\str{B}$ \emph{except} that the maximum element of
$\str{A}$ is identified with the minimum element of $\str{B}$.  The relation $\leq$
is interpreted in  $\str{A} \oplus \str{B}$ by taking the union of its
interpretations in $\str{A}$ and $\str{B}$ and letting $a \leq b$ for all
$a$ in $\str{A}$ and $b$ in $\str{B}$.  All other relation symbols are interpreted in  $\str{A} \oplus \str{B}$ by the union
of their interpretations in the two structures.  The operation of
ordered sum is clearly associative and we can thus write $\bigoplus_{i
  \in I} \str{A}_i$ for the ordered sum of a sequence of structures
indexed by an ordered set $I$.  Note that our use of the term ``ordered sum'' differs somewhat from its use, say by Ebbinghaus and Flum~\cite[Sec.~1.A.3]{EF99}.  The key difference is that in their definition we do not identify the maximum element of $\str{A}$ with the minimum element of $\str{B}$ but rather simply take the disjoint union of the two universes.

The structure $\Tot_{1,k}$ has $m=6(k+2)^2$ elements which we identify with the initial segment of the positive integers $[1,\ldots,m]$ with $\leq$ the natural linear order on these, $S$ the successor relation and the relation $R$ containing just the pair $(1,m)$.  The structure $\Gap_{1,k}$ is obtained from $\Tot_{1,k}$ by removing from the relation $S$ the \emph{central} pair of elements, i.e.\ $(m/2,m/2+1)$.

We now obtain $\str{N}_{1,k}$ as the ordered sum of $4(k+3)^3+2k+1$ copies of $\Gap_{1,k}$.  That is $\str{N}_{1,k} =  \bigoplus_{i \in [4(k+3)^3+2k+1]}\str{G}_i$ where each  $\str{G}_i$ is isomorphic to $\Gap_{1,k}$.  We also let $\str{M}_{1,k} = \bigoplus_{i \in [2(k+3)^3+k]}\str{G}_i \oplus \Tot_{1,k}\oplus \bigoplus_{i \in [2(k+3)^3+k]}\str{G}_i$.  In short, $\str{M}_{1,k}$ is obtained from $\str{N}_{1,k}$ by replacing the central copy of $\Gap_{1,k}$ with a copy of $\Tot_{1,k}$.

Let now $n \ge 2$ and suppose we have defined the $\sigma_{n-1}$-structures $\str{N}_{n-1,k},\str{M}_{n-1,k}, \Tot_{n-1,k}$ and $\Gap_{n-1,k}$.  Write $\str{N}_{n-1,k}^+$ and $\str{M}_{n-1,k}^+$ for the $\sigma_{n}$-structures that are obtained from $\str{N}_{n-1,k}$ and $\str{M}_{n-1,k}$ respectively by interpreting $P_{n}$ as the two element set $\{\mathrm{min},\mathrm{max}\}$ containing the minimum and maximum elements of the structure and $S_{n}$ as the relation containing the single pair $(\mathrm{min},\mathrm{max})$ ($R_{n}$ is empty in both these structures).
  Now, $\Tot_{n,k}$ is the structure obtained from $\bigoplus_{i\in [4(k+3)^{2n}+2k +1]}\str{M}_{n-1,k}^+$ (i.e.\ the ordered sum of $4(k+3)^{2n}+2k+1$ copies of $\str{M}_{n-1,k}^+$) by adding to the relation $R_{n}$ the pair relating the minimum and maximum elements of the linear order.  Similarly $\Gap_{n,k}$ is obtained from $\bigoplus_{i\in [2(k+3)^{2n}+k ]}\str{M}_{n-1,k}^+\oplus \str{N}_{n-1,k}^+ \oplus \bigoplus_{i\in [2(k+3)^{2n}+k ]}\str{M}_{n-1,k}^+$  by adding to the relation $R_{n}$ the pair relating the minimum and maximum elements of the linear order.  Equivalently, $\Gap_{n,k}$ is obtained from $\Tot_{n,k}$ by replacing the central copy of $\str{M}_{n-1,k}$ by a copy of $\str{N}_{n-1,k}$.

  Finally, we can define $\str{N}_{n,k}$ as the ordered sum of $4(k+3)^{2n+1}+2k+1$ copies of $\Gap_{n,k}$ and $\str{M}_{n,k}$ as the structure obtained from $\str{N}_{n,k}$ by replacing the central copy of $\Gap_{n,k}$ by a copy of $\Tot_{n,k}$.  
  This completes the definition of the structures.

  We now argue that for all values of $n$ and $k$, $\str{M}_{n,k}$ is a model of $\SomeTotalR_n$ and $\str{N}_{n,k}$ is not. 
  This is an easy induction on $n$.  For $n=1$, every interval $[x,y]$ of $\str{N}_{1,k}$ for which $R(x,y)$ holds induces a copy of $\Gap_{1,k}$.  By construction $S$ is not a complete successor relation in $\Gap_{1,k}$, and so $\str{N}_{1,k}$ does not satisfiy $\SomeTotalR_1$.  On the other hand, $\str{M}_{1,k}$ contains an interval $[x,y]$ with $R(x,y)$ that induces a copy of $\Tot_{1,k}$ and so $\str{M}_{1,k} \models \SomeTotalR_1$.

  Inductively, assume that $\str{M}_{n-1,k} \models \SomeTotalR_{n-1}$ and $\str{N}_{n-1,k} \not\models \SomeTotalR_{n-1}$.  Now, in both $\Tot_{n,k}$ and $\Gap_{n,k}$, the relation $S_n$ relates successive elements that are in $P_n$.  If $x,y$ is a pair of such successive elements then in $\Tot_{n,k}$ the interval $[x,y]$ always induces a structure whose $\sigma_{n-1}$-reduct is a copy of $\str{M}_{n-1,k}$ and therefore satisfies $\SomeTotalR_{n-1}$.  Hence $\Succ_n(x,y)$ is satisfied in $\Tot_{n,k}$ for all such pairs.  On the other hand, in $\Gap_{n,k}$ there is an interval $[x,y]$ with $S_n(x,y)$ which induces a structure whose $\sigma_{n-1}$-reduct is a copy of $\str{N}_{n-1,k}$ and therefore fails to satisfy $\SomeTotalR_{n-1}$.   Hence $\Total_n(x_0,y_0)$ is true in $\Tot_{n,k}$ and false in $\Gap_{n,k}$ when $x_0$ and $y_0$ are interpreted as the minimum and maximum elements in the structure respectively.  Since in $\str{N}_{n,k}$ all intervals $[x,y]$ for which $R_n(x,y)$ holds induce a copy of $\Gap_{n,k}$ and in $\str{M}_{n,k}$ there is such an interval which induces a copy of $\Tot_{n,k}$, we conclude that $\str{M}_{n,k} \models \SomeTotalR_{n}$ and $\str{N}_{n,k} \not\models \SomeTotalR_{n}$.

%% file: game.tex
\subsection{The Game Argument}\label{sec:game}

Our aim in this section is to establish the following lemma using an \efgame game argument:
\begin{lemma}\label{lem:main}
  For each $n,k$, $\str{N}_{n,k} \Rrightarrow_{2n+1,k} \str{M}_{n,k}$.
\end{lemma}

Our development of the Duplicator winning strategy in the game follows the inductive construction of the structures themselves.  For this, we first develop some tools for constructing strategies on ordered sums and expansions of structures from strategies on their component parts.  First, we introduce some useful notation.

For any ordered structure $\str{A}$, write $\str{A}^*$ for the expansion of $\str{A}$ with constants $\mathrm{min}$ and $\mathrm{max}$ interpreted by the minimum and maximum elements of the structure.  The main reason for introducing these is that we generally want to restrict attention to Duplicator strategies that respect the minimum and maximum elements and a notationally convenient way to do this is to have constants for these elements.

It is a standard fact that the equivalence relation $\equiv_m$ is a congruence with respect to various ways of combining structures.  In particular, it is so with respect to the notion of ordered sum defined by Ebbinghaus and Flum~\cite[Prop.~2.3.10]{EF99}.  The same method of composition of strategies can be used to show the following.
\begin{lemma}\label{lem:ordered-sum}
 If $\str{A}_1,\str{A}_2,\str{B}_1,\str{B}_2$ are ordered structures and $\bar{a}_1,\bar{a}_2,\bar{b}_1,\bar{b}_2$ tuples of elements from $\str{A}_1,\str{A}_2,\str{B}_1$ and $\str{B}_2$ respectively, such that $(\str{A}_1,\bar{a}_1)^* \Rrightarrow_{n,k} (\str{B}_1,\bar{b}_1)^*$ and $(\str{A}_2,\bar{a}_2)^* \Rrightarrow_{n,k} (\str{B}_2,\bar{b}_2)^*$, then
  $$(\str{A}_1 \oplus \str{A}_2,\bar{a}_1\bar{a}_2)^* \Rrightarrow_{n,k} (\str{B}_1 \oplus \str{B}_2,\bar{b}_1\bar{b}_2)^*.$$
\end{lemma}
Note that it is an immediate consequence that the same is true with $\equiv_{n,k}$ in place of $\Rrightarrow_{n,k}$.

Furthermore, this also extends to ordered sums of sequences.  Moreover, we do not have to match the lengths of the sequences as long as they are long enough.  Again, this is standard for the equivalence relation $\equiv_m$~\cite[Ex.~2.3.13]{EF99}, for sequences of length at least $2^m$ and a slightly different notion of ordered sum.  For our relations we obtain a tighter bound, so we prove this explicitly as the proof is instructive.

Define the following function $\rho$ on pairs of natural numbers by recursion.
\begin{eqnarray*}
  \rho(1,k) & = &  2k+2\\
  \rho(n+1,k) & = & (k+2)(\rho(n,k) + 1).
\end{eqnarray*}
A simple induction on $n$ shows that $2(k+3)^n > \rho(n,k)$ for all $k,n \ge 1$.

\begin{lemma}\label{lem:sequence-sum}
  If $\str{A}$ and $\str{B}$ are ordered structures, 
$\str{A}^* \Rrightarrow_{n, k} \str{B}^*$ and  $s, t \ge \rho(n,k)$, then
$\big(\bigoplus_{1\leq i\leq s} \str{A}_i\big)^* \Rrightarrow_{n, k}
\big(\bigoplus_{1\leq j \leq t}\str{B}_i\big)^*$, where $\str{A}_i \equiv_{n,k} \str{A}$ and $\str{B}_i \equiv_{n,k} \str{B}$ for all $i$.
\end{lemma}
\begin{proof}
  The proof is by induction on $n$.  Suppose $n=1$ and Spoiler plays a move choosing $k$ elements from $\bigoplus_{1\leq i\leq s} \str{A}_i$.  Suppose these are chosen from $\str{A}_{i_1},\ldots,\str{A}_{i_l}$ with $1 \leq i_1 < \cdots < i_l \leq s$ for some $l\leq k$.  Further, let $i_0 = 0$ and $i_{l+1} = t$.  Since $l+2 \leq k+2 < 2k+2 = \rho(1,k) \leq  s$, there is some $p$ such that $i_{p+1} > i_p+1$.  Duplicator must choose structures $(\str{B}_{j_q})_{1\leq q \leq l}$ in which to respond.  Moreover, since $\str{A}_i$ and $\str{A}_{i+1}$ share an element for all $i$, whenever $i_{q+1} = i_q + 1$, we must choose $j_{q+1} = j_q +1$.
  
  Duplicator chooses values $ 0 = j_0 <  j_1 < \cdots < j_l \leq j_{l+1} = t$ as follows.  
For all values of $q$ from $0$ to $p-1$, choose $j_{q+1} = j_q + 1$ if $i_{q+1} = i_q + 1$ and choose $j_{q+1} = j_{q}+2$ otherwise.   For all values of $q$ from $l$ down to $p+1$, choose $j_q = j_{q+1} - 1$ if $i_q = i_{q+1} -1$  and choose $j_q = j_{q+1} -2$ otherwise.  Because $t \ge 2k+2$, this guarantees that $j_{p+1} > j_p + 1$.  Thus, Duplicator can 
respond to the elements picked in $\str{A}_{i_p}$ with elements in $\str{B}_{j_p}$ by composing the winning strategies for $\str{A}_{i_p} \equiv_{1,k} \str{A} \Rrightarrow_{1,k} \str{B} \equiv_{1,k} \str{B}_{j_p}$ and this is a winning response.  Moreover, by construction, this strategy maps the minimum and maximum elements of $\bigoplus_{1\leq i\leq s} \str{A}_i$ to the corresponding elements of $\bigoplus_{1\leq i\leq t} \str{B}_j$.

  Now suppose $n \ge 2$ and the statement has been proved for $n-1$.  Let Spoiler play a move choosing $k$ elements from $\bigoplus_{1\leq i\leq s} \str{A}_i$, and again say these are chosen from $\str{A}_{i_1},\ldots,\str{A}_{i_l}$ with $1 \leq i_1 < \cdots < i_l \leq s$ for some $l\leq k$.  Further, define $i_0 = 0$ and $i_{l+1} = s$.  Since $l+2 \leq k+2$ and $t \ge \rho(n,k) = (k+2)(\rho(n-1,k)+1)$, Duplicator can choose indices  $0=j_0 < j_1 \cdots < j_l < j_{l+1} = t$ so that for all $p$ either $j_{p+1} -j_p = i_{p+1} - i_p$ or $(j_{p+1} -j_p),(i_{p+1} - i_p) \geq \rho(n-1,k)+1$.  Now choose, for each $p$ with $1 \leq p \leq l$ a response $\bar{b}_p$ for Duplicator in $\str{B}_{j_p}$ to the elements $\bar{a}_p$ chosen by Spoiler in $\str{A}_{i_p}$.  We claim that
  $$ \big(\bigoplus_{1\leq j \leq t} \str{B}_j, (\bar{b}_p)_{1\leq p\leq l} \big)^* \Rrightarrow_{n-1,k} \big(\bigoplus_{1\leq i \leq s}\str{A}_i, (\bar{a}_p)_{1\leq p\leq l} \big)^*.$$

  To prove this, it suffices to show for each $p$ with $0 \leq p \leq l$ that
  \begin{claim}\label{clm:sequence-sum}
    \[
      \big(\bigoplus_{j_{p}< j \leq j_{p+1}}(\str{B}_j, \bar{b}_{p})\big)^* \Rrightarrow_{n-1,k} \big(\bigoplus_{i_{p}< i \leq i_{p+1}}(\str{A}_i, \bar{a}_p) \big)^*,
      \]
  \end{claim}
  for then the claim follows by $l$ applications of Lemma~\ref{lem:ordered-sum}.  Note that we have,
  \begin{enumerate}
  \item   for all $i,j$ that $\str{B}_j \Rrightarrow_{n-1,k} \str{A}_i$ by the assumption that $\str{A} \Rrightarrow_{n,k} \str{B}$; and \label{first}
  \item for all $p$ we have $(\str{B}_{j_p},\bar{b}_p) \Rrightarrow_{n-1,k} (\str{A}_{i_p},\bar{a}_p)$ by the choice of $\bar{b}_p$ as Duplicator's winning response to Spoiler's choice of $\bar{a}_p$. \label{second}
  \end{enumerate}
Now, for each value of $p$ there are two possibilities:
\begin{description}
\item[case (i):]  $j_{p+1} -j_p = i_{p+1} - i_p$.  In this case, the two sides of Claim~\ref{clm:sequence-sum} are the ordered sums of sequences of equal length.  The corresponding pieces are all related by $\Rrightarrow_{n-1,k}$, either by ~\ref{first}.\ above for all except the last piece or by~\ref{second}.\ for the last piece.  Thus, Claim~\ref{clm:sequence-sum} is established by application of Lemma~\ref{lem:ordered-sum}; 
\item[case (ii):] $(j_{p+1} -j_p),(i_{p+1} - i_p) \geq \rho(n-1,k)+1$.  In this case, the structures on the two sides of Claim~\ref{clm:sequence-sum} can be expressed as $ \big(\bigoplus_{j_{p}< j \leq j_{p+1} - 1}\str{B}_j\big) \oplus (\str{B}_{j_{p+1}}, \bar{b}_{p+1}) $ and $\big(\bigoplus_{i_{p}< i \leq i_{p+1} - 1}\str{A}_i\big) \oplus (\str{A}_{i_{p+1}}, \bar{a}_{p+1}) $, respectively.  Since $(j_{p+1} -j_p -1),(i_{p+1} - i_p -1) \geq \rho(n-1,k)$, we have by the induction hypothesis and~\ref{first}.\ that $ \big(\bigoplus_{j_{p}< j \leq j_{p+1} - 1}\str{B}_j\big)^* \Rrightarrow_{n-1,k} \big(\bigoplus_{i_{p}< i \leq i_{p+1} - 1}\str{A}_i\big)^*$.  This, together with~\ref{second}.\ and Lemma~\ref{lem:ordered-sum} establishes Claim~\ref{clm:sequence-sum}, and hence the inductive step of the proof.
\end{description}
\end{proof}

Besides ordered sums, another key step in the inductive constructions of our structures is adding unary relations which include the minimum and maximum elements of a structure and adding binary relations which relate the minimum and maximum elements.  These operations also behave well with respect to games.  To be precise, suppose $U$ is a unary relation symbol and $T$ a binary relation symbol.  Let $\str{A}$ be an ordered structure with minimum and maximum elements $\mathrm{min}$ and $\mathrm{max}$ respectively. Write $\str{A}_U$ for the structure obtained from $\str{A}$ by including $\mathrm{min}$ and $\mathrm{max}$ in the interpretation of $U$.   Similarly, write $\str{A}_T$ for the structure obtained from $\str{A}$ by adding the pair $(\mathrm{min},\mathrm{max})$ to the interpretation of $T$.  Note that we do not assume that $U$ or $T$ are in the vocabulary of $\str{A}$.  If they are not, then their interpretations in $\str{A}_U$ and $\str{A}_T$ respectively contain nothing other than the elements added.
\begin{lemma}\label{lem:min-max}
  Let $n,k \ge q$.  If $\str{A}$ and $\str{B}$ are ordered structures for which $\str{A}^* \Rrightarrow_{n,k} \str{B}^*$ then $\str{A}_U \Rrightarrow_{n,k} \str{B}_U$ and $\str{A}_T \Rrightarrow_{n,k}\str{B}_T$.
\end{lemma}
\begin{proof}
  This is immediate from the fact that a Duplicator winning strategy between $\str{A}^*$ and $\str{B}^*$ must map the minimum elements of the two structures to each other, and similarly for the maximum.
\end{proof}

We are now ready to start inductively constructing the Duplicator winning strategy that establishes Lemma~\ref{lem:main}.  We begin with games on some simple structures.  For any $m \ge 1$ write $L_m$ for the structure with exactly $m$ elements and two binary relations $\leq$ and $S$ where $\leq$ is a linear order and $S$ the corresponding successor relation.
\begin{lemma}\label{lem:successor}
  If $m_1,m_2 > \rho(n,k)$ then $L_{m_1}^* \Rrightarrow_{n,k} L_{m_2}^*$.
\end{lemma}
\begin{proof}
  Note that $L_m$ is the ordered sum of a sequence of $m-1$ copies of $L_2$, so the result follows immediately from Lemma~\ref{lem:sequence-sum}.
\end{proof}

Without loss of generality, assume that the universe of $L_m$ is $\{1,\ldots,m\}$ and write $G_m$ for the structure obtained from $L_{m+1}$ by deleting the element $\lceil \frac{m}{2}\rceil$.  Note that $G_m$ is isomorphic to the structure obtained from $L_{m}$ by removing from the relation $S$ the pair $(\lceil \frac{m}{2}\rceil -1, \lceil \frac{m}{2}\rceil)$.

\begin{lemma}\label{lem:succ-gap}
  If $m_1,m_2 \ge 2k+2$ then $G_{m_1}^* \Rrightarrow_{1,k} L_{m_2}^*$.
\end{lemma}
\begin{proof}
  Since $G_{m_1}^*$ is a substructure of $L_{m_1+1}^*$, every existential sentence true in the former is also satisfied in the latter.  Now, since $L_{m_1+1}^* \Rrightarrow_{1,k} L_{m_2}^*$ by Lemma~\ref{lem:successor}, the result follows.
\end{proof}

The next two lemmas give us the base case of the inductive proof of Lemma~\ref{lem:main}.

\begin{lemma}\label{lem:TGbase}
  $\Tot_{1,k} \Rrightarrow_{2,k} \Gap_{1,k}$
\end{lemma}
\begin{proof}
  Recall that the $\{\leq,S\}$-reduct of $\Tot_{1,k}$ is the structure $L_m$ for $m=6(k+2)^2$.  Note that $m > 2\rho(2,k)+(k+2)(2k+4)$.  We think of this as composed of  three segments: the first $\rho(2,k)$ elements; the last $\rho(2,k)$ elements and a \emph{middle segment} containing the remainder.
  Suppose now that Spoiler  chooses $k$ elements $a_1 < \cdots < a_k$ from $\Tot_{1,k}$ in the first round of the game.  Since the middle segment  contains more than $(k+2)(2k+4)$ elements, it must contain an interval $[x,y]$ of  $2k+2$ consecutive elements which are not chosen.  Let us say that $a_i < x$ and $y < a_{i+1}$.  Thus, we can write the $\{\leq,S\}$-reduct of $\Tot_{1,k}$ with the chosen elements as
  $$ (L_{m_1},a_1,\ldots,a_i) \oplus L_{m_2} \oplus  (L_{m_3},a_{i+1},\ldots,a_k),$$
  where $m_1,m_3 \geq \rho(2,k)$ and $m_2 \geq 2k+2$.

  Consider now $\Gap_{1,k}$.  The $\{\leq,S\}$-reduct of this structure is $G_m$, which can also be written as $L_{\frac{m}{2}-k-1}\oplus G_{2k+2} \oplus L_{\frac{m}{2}-k-1}$.  Since $\frac{m}{2}-k-1 > \rho(2,k)$, we have by Lemma~\ref{lem:successor} that $L_{m_1}^* \Rrightarrow_{2,k}  L_{\frac{m}{2}-k-1}^*$ and hence there is a choice of elements $b_1,\ldots,b_i$ such that $(L_{\frac{m}{2}-k-1},b_1,\ldots,b_i)^* \Rrightarrow_{1,k} (L_{m_1},a_1,\ldots,a_i)^*$.  Similarly, there is a choice of elements $b_{i+1},\ldots,b_k$ such that $(L_{\frac{m}{2}-k-1},b_{i+1},\ldots,b_k)^* \Rrightarrow_{1,k} (L_{m_3},a_{i+1},\ldots,a_k)^*$.  Further, we know that $G_{2k+2}^* \Rrightarrow_{1,k} L_{m_3}^*$ by Lemma~\ref{lem:succ-gap}.  Hence, by Lemma~\ref{lem:ordered-sum} we have that $L_m^* \Rrightarrow_{2,k} G_m^*$.  The result now follows by Lemma~\ref{lem:min-max} as $\Tot_{1,k}$ and $\Gap_{1,k}$ are obtained from $L_m$ and $G_m$ respectively by relating the minimum and maximum elements with the relation $R$.  
\end{proof}

A similar pattern of argument is repeated in the second base case, and we will be less detailed in spelling it out.

\begin{lemma}\label{lem:MNbase}
  $\str{N}_{1,k}^* \Rrightarrow_{3,k} \str{M}_{1,k}^*$.
\end{lemma}
\begin{proof}
  Recall that $\str{N}_{1,k}$ is the ordered sum of $m=4(k+3)^3+2k+1$ copies of $\Gap_{1,k}$.  So $\str{N}_{1,k} =  \bigoplus_{i \in [m]}\str{G}_i$.  Note that $m > 2\rho(3,k) + k+1$.   Suppose now that Spoiler  chooses $k$ elements $a_1 < \cdots < a_k$ in the first round of the game.  Thus, there is an index $i$ in the middle segment of $[m]$ of length $k+1$ such that $\str{G}_i$ does not contain a chosen element and we can write $\str{N}_{1,k}$ with the chosen elements as $(\bigoplus_{i \in [m_1]}\str{G}_i,a_1,\ldots,a_j) \oplus \Gap_{1,k} \oplus (\bigoplus_{i \in [m_2]}\str{G}_i,a_{j+1},\ldots,a_k)$, where $m_1,m_2 > \rho(3,k)$.

  On the other side,  $\str{M}_{1,k} =  \bigoplus_{i \in [(m-1)/2]}\str{G}_i \oplus \Tot_{1,k} \oplus  \bigoplus_{i \in [(m-1)/2]}\str{G}_i$.  Since $(m-1)/2 > \rho(3,k)$, by Lemma~\ref{lem:sequence-sum} we have $ \big(\bigoplus_{i \in [m_1]}\str{G}_i\big)^* \Rrightarrow_{3,k} \big( \bigoplus_{i \in [(m-1)/2]}\str{G}_i\big)^*$ and $ \big(\bigoplus_{i \in [m_2]}\str{G}_i\big)^* \Rrightarrow_{3,k} \big( \bigoplus_{i \in [(m-1)/2]}\str{G}_i\big)^*$.  Thus, we can find elements $b_1,\ldots,b_k$ such that
  $$\big( \bigoplus_{i \in [\frac{m-1}{2}]}\str{G}_i,b_1,\ldots,b_j\big)^* \Rrightarrow_{2,k} (\bigoplus_{i \in [m_1]}\str{G}_i,a_1,\ldots,a_j)^*$$ and
  $$\big( \bigoplus_{i \in [\frac{m-1}{2}]}\str{G}_i,b_{j+1},\ldots,b_k\big)^* \Rrightarrow_{2,k} (\bigoplus_{i \in [m_2]}\str{G}_i,a_{j+1},\ldots,a_k)^*.$$  Combining this with the fact that $\Tot_{1,k} \Rrightarrow_{2,k} \Gap_{1,k}$ by Lemma~\ref{lem:TGbase}, we get by Lemma~\ref{lem:ordered-sum} that $(\str{M}_{1,k},b_1,\ldots,b_k)^* \Rrightarrow_{2,k} (\str{N}_{1,k},a_1,\ldots,a_k)^*$ and the result follows.
\end{proof}

These last two lemmas form the base case of the
induction that establishes the main result.  Where the argument is analogous to the previous ones, we skim over the details.

\begin{proof}[Proof of Lemma~\ref{lem:main}]
  We prove the following two statements by induction for all $n,k \geq 1$.
  \begin{enumerate}
  \item $\Tot_{n,k} \Rrightarrow_{2n,k} \Gap_{n,k}$
  \item $\str{N}_{n,k}^* \Rrightarrow_{2n+1,k} \str{M}_{n,k}^*$
  \end{enumerate}
  The case of $n=1$ is established in Lemmas~\ref{lem:TGbase} and~\ref{lem:MNbase} respectively.  Suppose now $n \ge 2$ and we have established both statements for $n-1$.  

  First, recall that $\str{M}_{n-1,k}^+$ and $\str{N}_{n-1,k}^+$ are obtained from $\str{M}_{n-1,k}$ and $\str{N}_{n-1,k}$ respectively by including their minimum and maximum elements in the unary relation $P_n$ and the binary relation $S_n$.  Thus, by the induction hypothesis and Lemma~\ref{lem:min-max}, we have $\str{N}_{n-1,k}^+ \Rrightarrow_{2n-1,k} \str{M}_{n-1,k}^+$.

  $\Tot_{n,k}$ consists of the ordered sum of a sequence of $m = 4(k+3)^{2n}+2k+1 > 2\rho(2n,k) + k + 1$ copies of $\str{M}_{n-1,k}^+$, along with a relation $R_n$ containing just the pair with the minumum and maximum elements.  When Spoiler chooses $k$ elements from this structure, we can find a copy of $\str{M}_{n-1,k}^+$ in the middle $k+1$ copies that has no element chosen and there are $m_1 > \rho(2n,k)$ copies before it and $m_2 > \rho(2n,k)$ copies after it.  We can similarly express $\Gap_{n,k}$ as the ordered sum of a sequence of $(m-1)/2 > \rho(2n,k)$ copies of $\str{M}_{n-1,k}^+$, followed by a copy of $\str{N}_{n-1,k}^+$ and a further $(m-1)/2 > \rho(2n,k)$ copies of $\str{M}_{n-1,k}^+$.  Lemma~\ref{lem:sequence-sum} tells us that we can find a response to the chosen elements in the first and third parts.  This combined with the fact that $\str{N}_{n-1,k}^+ \Rrightarrow_{2n-1,k} \str{M}_{n-1,k}^+$  and using Lemma~\ref{lem:min-max} to expand to the relation $R_n$   gives us the desired result.

  The argument for the second statement is entirely analogous.  $\str{N}_{n,k}$ is the ordered sum of a sequence of $m = 4(k+3)^{2n+1}+2k+1 > 2\rho(2n+1,k) + k + 1$ copies of $\Gap_{n,k}$.  When Spoiler chooses $k$ elements from this structure, we can find a copy of $\Gap_{n,k}$ in the middle $k+1$ copies that has no element chosen and there are $m_1 > \rho(2n+1,k)$ copies before it and $m_2 > \rho(2n+1,k)$ copies after it.  Since $\str{M}_{n,k}$ has $(m-1)/2 > \rho(2n+1,k)$ copies of $\Gap_{n,k}$ followed by a copy of $\Tot_{n,k}$ and a further $(m-1)/2 > \rho(2n+1,k)$ copies of $\Gap_{n,k}$,  by Lemma~\ref{lem:sequence-sum} we can find responses to the chosen elements in the first and third parts. We have already proved that $\Tot_{n,k} \Rrightarrow_{2n,k}  \Gap_{n,k}$.  We can combine these to complete the proof.  
\end{proof}

%% file: conclusion.tex
\section{Concluding Remarks}
We have established in this paper that the extension-closed properties of finite structures that are definable in first-order logic are not contained in any fixed quantifier-alternation fragment of the logic.  The construction of the sentences demonstrating this is recursive.  It builds on a base of known counter-examples for the \LT theorem in the finite and lifts them up inductively.  Also, the argument for showing inexpressibility in fixed levels of the quantifier-alternation hierarchy builds on the game arguments used with previously known examples and builds on them systematically using a form of Feferman-Vaught decomposition (see~\cite{makowsky}) for ordered sums of structures.

Our result actually establishes that for all \emph{odd} $n > 1$, there is a $\Sigma_n$-definable property that is closed under extensions but not definable by a $\Pi_n$ sentence.  Interestingly, this is not true for even values of $n$.  In particular, it is known that every $\Sigma_2$-definable extension-closed property is already definable in $\Sigma_1$.  This observation is credited to Compton in~\cite{gurevich} and may also be found in~\cite{SankaranAMKC12}.  We do not know, however, whether for even $n > 2$, the extension closed properties in $\Sigma_n$ can all be expressed in $\Pi_n$ or even $\Sigma_{n-1}$.  On the other hand, we are able to observe that for all even $n > 1$, there is a $\Pi_n$-definable extension-closed property that is not in $\Sigma_n$.  This is a direct consequence of our proof.  Indeed, consider the sentence $\phi_n$ obtained from $\SomeTotalR_n$ by removing the two outer existential quantifiers and replacing the resulting free variables with new constants $a$ and $b$.  This is a $\Pi_{2n}$ sentence and is easily seen to define an extension-closed class.  By our construction, $\phi_n$ is satisfied in the expansion of $\Tot_{n,k}$ where $a$ and $b$ are the minimum and maximum elements respectively.  At the same time $\phi_n$ is false in the similar expansion of $\Gap_{n,k}$.  Since we showed that $\Tot_{n,k}^* \Rrightarrow_{2n,k} \Gap_{n,k}^*$, it follows that $\phi_n$ is not equivalent to a $\Sigma_{2n}$ sentence.  It would be interesting to complete the picture of extension-closed properties in the remaining quantifier-alternation fragments, specifically $\Sigma_n$ for even values of $n$ and $\Pi_n$ for odd values of $n$.

It is a feature of our construction that the vocabulary $\sigma_n$ in which we construct the sentences which separate extension-closed $\Sigma_{2n+1}$ from $\Pi_{2n+1}$ grows with $n$.  Could our results be established in a fixed vocabulary?  Indeed, does something like Theorem~\ref{thm:main} hold for finite graphs?

Another interesting direction left open from our work is the relation with $\Datalog$.  All the extension-closed $\fo$-definable properties we construct are also definable in $\Datalog$.  Rosen and Weinstein~\cite{RosenW} ask whether this is true for all $\fo$-definable extension-closed properties, and this remains open.  Indeed, it is conceivable that we have an extension-preservation theorem for least fixed-point logic in the finite, so that even all LFP-definable extension-closed properties are in $\Datalog$.